\documentclass[a4paper,fleqn,11pt]{article}
\usepackage[tbtags]{amsmath}
\usepackage{amsfonts}
\usepackage{amsthm}
\usepackage{color}
\usepackage{amsfonts}
\usepackage{multirow}
\usepackage{graphicx,amssymb}
\usepackage{rotating}
\usepackage{booktabs}
\usepackage{caption}
\usepackage{bbm}
\usepackage{subfig}
\usepackage{xcolor}
\usepackage{authblk}
\usepackage{natbib}[round]
\usepackage[linesnumbered,ruled,lined,boxed]{algorithm2e}

\newtheorem{proposition}{\bf Proposition}[section]

\newcommand\bbone{\ensuremath{\mathbbm{1}}}
\def\bomega{\mbox{\boldmath $\omega$}} 
\def\bpsi{\mbox{\boldmath $\psi$}}
\def\br{\mathbf{r}}
\def\by{\mathbf{y}} 
\def\bs{\mathbf{s}}
\def\bS{\mathbf{S}}
\def\bI{\mathbf{I}}
\def\bz{\mathbf{z}}

\def\bc{\mathbf{c}} 
\def\bx{\mathbf{x}}
\def\bw{\mathbf{w}}%
\def\bA{\mathbf{A}}

\begin{document}

\title{\bf\textsf{Numerical evaluation of dual norms via the MM algorithm}}
\author[1]{Mauro Bernardi\thanks{Corresponding author: Via C. Battisti, 241, 35121 Padua, Italy. e-mail: mauro.bernardi@unipd.it, Phone.: +39.049.8274165.}}
\author[2]{Marco Stefanucci}
\author[3]{Antonio Canale}
\affil[1]{Department of Statistical Sciences, University of Padova and Institute for Applied Mathematics ``Mauro Picone'', IAC-CNR, Rome, Italy}
\affil[2]{Dipartimento di Scienze Statistiche, Universit\`a degli Studi di Roma - La Sapienza, Roma, Italy}
\affil[3]{Department of Statistical Sciences, University of Padova, Padova, Italy}

\date{\today}

\maketitle
\begin{abstract}
\noindent We deal with the problem of numerically computing the dual norm, which is important to study sparsity-inducing regularizations \citep[][]{jenatton2011structured,bach_etal.2012}. The dual norms find application in optimization and statistical learning, for example, in the design of working-set strategies, for characterizing dual gradient methods, for dual decompositions and in the definition of augmented Lagrangian functions. Nevertheless, the dual norm of some well-known sparsity-inducing regolarization methods are not analytically available. Examples are the overlap group $\ell_2$-norm of \cite{jenatton2011structured} and the elastic net norm of \cite{zou_hastie.2005}. Therefore we resort to the Majorization-Minimization principle of \cite{lange.2016} to provide an efficient algorithm that leverages a reparametrization of the dual constrained optimization problem as unconstrained optimization with barrier. Extensive simulation experiments have been performed in order to verify the correctness of operation, and evaluate the performance of the proposed method. Our results demonstrate the effectiveness of the algorithm in retrieving the dual norm even for large dimensions.
\end{abstract}
%
%

%
\section{Introduction}
\label{sec:intro}
%
\noindent In this paper we deal with the problem of computing the dual norm, which is important to
study sparsity-inducing regularizations \citep[][]{jenatton2011structured,bach_etal.2012}. 
For any vector $\bx=(x_1,\dots,x_p)^\intercal\in\mathbb{R}^p$ the dual norm $\Omega^\ast(\bx)$ of the norm $\Omega(\bx)\in\mathbb{R}^+$ (defined over the euclidean space $\mathcal{X}$) is defined as $\Omega^\ast(\bx)=\sup_{\bz\in\mathbb{R}^p}\big\{\bz^\intercal\bx\,\vert\,\Omega(\bz)\leq1\big\}$. The previous optimization problem can be equivalently rewritten as
\begin{align}
\label{eq:dual_norm_opt_prob_def_1}
&\Omega^\ast(\bx)=\inf_{\bz\in\mathbb{R}^p}-\bz^\intercal\bx\\
\label{eq:dual_norm_opt_prob_def_2}
&\upsilon(\bz)\geq 0,
\end{align}
where $\upsilon(\bz)=1-\Omega(\bz)\in(0,1)$. Moreover, the dual norm of $\Omega^\ast(\bx)$ is $\Omega(\bx)$ itself, and as a consequence, the formula above holds also if the roles of $\Omega(\bx)$ and $\Omega^\ast(\bx)$ are exchanged. 
The dual norms find application in several contexts: in the design of working-set strategies, for characterizing dual gradient methods, for dual decompositions and in the definition of augmented Lagrangian functions that arises in the context of optimization methods \citep[see, e.g.][]{boyd_etal.2011}. Two further interesting applications of dual norms are to verify if the irrepresentable condition holds in the context of model selection consistency of sparsity-inducing penalties, \citep[see, e.g.][]{buhlmann_van_de_geer.2011} and for computing optimality conditions of sparse regularized problems \citep[see, e.g.][]{osborne_etal.2000}. In the latter case, sparse regularized problems involve convex optimizations of the form 
\begin{equation}
\label{eq:sparse_prob}
\max_{\bx\in\mathbb{R}^p}h(\bx)+\lambda\Omega(\bx),
\end{equation}
for any $\lambda>0$, where $f:\mathbb{R}^p\to\mathbb{R}$ is a convex differentiable function and $\Omega:\mathbb{R}^p\to\mathbb{R}$ is a
sparsity-inducing---typically nonsmooth and non-euclidean---norm. Simple calculations show that $\bx^\star\in\mathbb{R}^p$ is optimal for \eqref{eq:sparse_prob} if and only if $\lambda^{-1}\nabla f(\bx^\star)\in\partial\Omega(\bx^\star)$ where
\begin{equation}
\partial\Omega(\bx^\star)=\begin{cases}
\{\bomega\in\mathbb{R}^p;\,\Omega^\ast(\bomega)\leq 1\};&\mathrm{if}\quad\bx^\star=0\\
\{\bomega\in\mathbb{R}^p;\,\Omega^\ast(\bomega)\leq 1,\,\&\,\bomega^\intercal\bx^\star=\Omega(\bx^\star)\};&\mathrm{if}\quad\bx^\star\neq0.
\end{cases}
\end{equation}
From previous conditions it is possible to derive  interesting properties of the problem \eqref{eq:sparse_prob}, as well as efficient algorithms for solving it.
Consider now the problem of checking the irrepresentable condition for \eqref{eq:sparse_prob} where, for example, $h(\bx)=\Vert \bA\bx-\by\Vert_2$ and $\bA\in\mathbb{R}^{n\times p}$ is a given design matrix and $\by\in\mathbb{R}^n$ is known. We need to calculate the dual norm $\Omega^\ast\big( {\mathbf{A}}_0^\intercal {\mathbf{A}}_1 ({\mathbf{A}}_1^\intercal {\mathbf{A}}_1)^{-1} \mathbf{r}_1  \big)$, where $\bA_0$ and $\bA_1$ are partitions of the design matrix that corresponds to the true zeros and non-zeros of $\bx$, respectively, and $\br_1$ is the the first derivative of the norm $\Omega(\bx)$, $\Omega^\prime(\bx)$, evaluated at the optimal solution of \eqref{eq:sparse_prob}, $\bx^\star$. The irrepresentable condition is important to study sparse-inducing norms for model selection purposes, \citep[see, e.g.][for an overview]{zhao_yu.2006,buhlmann_van_de_geer.2011}.\par
Special cases of sparsity-inducing norms are, the $\ell_1$-norm \citep[][]{tibshirani1996} $\Omega_1(\bx)=\Vert\bx\Vert_1=\sum_{j=1}^p\vert x_j\vert$, the $\ell_\infty$-norm $\Omega_\infty(\bx)=\max_{j=1,\dots,p}\vert x_j\vert$, the $\ell_2$-norm \citep[][]{hoerl_kennard.2000}, $\Omega_2(\bx)=\Vert\bx\Vert_2$, the Group $\ell_2$-norm \citep[][]{yuan_lin.2006} $\Omega_G(\bx)=\sum_{g=1}^G\sqrt{w_g}\Vert \bx_g\Vert_2$, where $\bx=\{\bx_{(1)},\dots,\bx_{(G)}\}$ and each $\bx_{(g)}$ for $1\leq g\leq G$ represents a group from $\bx$ and $w_g$ represents the number of weights in $\bx_g$ for $1\leq g\leq G$. In all those cases an analytical solution for the dual norm is available. In particular, the $\ell_1$- and $\ell_\infty$-norms are dual to each other, and the $\ell_2$-norm is self-dual (dual to itself). As concerns the Group $\ell_2$-norm $\Omega_G(\bx)$ we have instead that $\Omega^\ast_G(\bx)=\max_{g=1,\dots,G}\frac{1}{\sqrt{w_g}}\Vert \bx_g\Vert_2$. However, for the overlap group $\ell_2$-norm, introduced by \cite{jenatton2011structured} as a relaxation of the group $\ell_2$-norm that allows for groups to share the same components of the vector $\bx$, an analytical solution for the dual norm does not exists. The overlap group LASSO norm finds application in many relevant applied contexts of statistical learning to induce sparse solutions, \citep[see][for extensive examples]{jenatton2011structured}. When an analytical solution to the dual norm does not exist we can retrieve optimization methods for solving problem \eqref{eq:dual_norm_opt_prob_def_1}-\eqref{eq:dual_norm_opt_prob_def_2}. However, the computational cost of performing the constrained optimziation in \eqref{eq:dual_norm_opt_prob_def_1}-\eqref{eq:dual_norm_opt_prob_def_2} is especially high when $p$ is large, \citep[see][]{jenatton2011structured}.\par
Previous considerations motivate our idea to provide an efficient algorithm for computing the dual norm for the case where an analytical solution does not exist or direct optimization of \eqref{eq:dual_norm_opt_prob_def_1}-\eqref{eq:dual_norm_opt_prob_def_2} is not feasible. The proposed algorithm relies on the Majorization-Minimization (MM, hereafter) principle of \cite{lange.2013,lange.2016} for dealing with the constrained optimization problem with penalties or barriers.\par
The remaining of the paper is organized as follows. Section \ref{sec:barrier_methods} introduces barrier methods for efficiently dealing with constrained optimization problems and presents the usual Newton-Raphson solution. Section \ref{sec:mm_algo} introduces the MM algorithm and Section \ref{sec:appl} applies the MM algorithm to the overlap group $\ell_2$-norm of \cite{jenatton2011structured}. Section \ref{sec:conclu} concludes.
%
\section{Barrier methods for constrained optimization}
\label{sec:barrier_methods}
%
\noindent Let us consider the constrained optimization problem 
\begin{equation}
\label{eq:copt}
\Omega^\ast(\bx)=\inf_{\bz\in\mathbb{R}^p}-\bz^\intercal\bx,\quad\text{s.t.}\quad\upsilon(\bz)\geq 0,
\end{equation}
where $\mathbf{z}\in\mathbb{R}^p$, $\mathbf{x}\in\mathbb{R}^p$ and $\upsilon(\bz)=1-\Omega(\bz)\in(0,1)$. 
Optimization \eqref{eq:copt} falls into the class of linear programming problems subject to nonlinear constraints, for which there exists lots of efficient solvers \citep[See, e.g.][]{nocedal_wright.2006}. However, the huge number of constraints may be a crucial issue that even the most sophisticated tools in the linear programming literature are not able to manage easily. Here, we adopt a different strategy, avoiding reparametrizations and transforming \eqref{eq:copt} onto an unconstrained problem \citep[see][]{luenberger_yinyu.2021}. Now, let us define the Lagrangian function 
\begin{equation}
\label{eq:L}
\mathcal{L}_\varrho(\bz) \equiv \mathcal{G}(\bz) + \varrho \mathcal{J}(\bz),
\end{equation}
as the sum of the objective function $\mathcal{G}(\bz)$ and of the penalization term $\mathcal{J}(\bz)$ weighted by a positive penalization parameter $\varrho>0$. And, in particular, $\mathcal{J}(\bz)$ is chosen to be the negative of the logarithmic transformation of the constraint violation, namely:
\begin{equation}
\label{eq:J}
\mathcal{J}(\bz) = -\log(\upsilon(\bz)).
\end{equation}
Then, assuming that $\varrho$ is big enough and requiring other technical but mild conditions 
\citep[Theorem 17.3, Page 507]{nocedal_wright.2006}, the constrained problem \eqref{eq:copt} shares the 
same solution of the penalized unrestricted optimization of the Lagrangian function, 
that is $\widehat{\bz} = \arg\min_{\bz} \mathcal{L}_\varrho(\bz)$. Therefore, the unconstrained optimization problem becomes
\begin{align}
\Omega^\ast(\bx)&=\inf_{\bz\in\mathbb{R}^p}\mathcal{L}_\varrho(\bz)\\
\mathcal{L}_\varrho(\bz)&=\mathcal{G}(\bz)-\varrho\log(\upsilon(\bz)),
\end{align}
for $\varrho\in\mathbb{R}^+$. 
%
\subsection{Newton-Raphson solution}
\label{sec:newton_raphson_sol}
%
\noindent Newton-Raphson-type algorithms require the evaluation of the gradient and the hessian matrix of the objective function 
\begin{align}
\mathcal{L}_\varrho^\prime(\bz)&=-\bx-\varrho\frac{\upsilon^\prime(\bz)}{\upsilon(\bz)}\\
\mathcal{L}_\varrho^{\prime\prime}(\bz)&=-\varrho\frac{\upsilon^{\prime\prime}(\bz)}{\upsilon(\bz)}+\varrho\frac{\upsilon^{\prime}(\bz)\big(\upsilon^{\prime}(\bz)\big)^\intercal}{(\upsilon(\bz))^2},
\end{align}
where $\upsilon^{\prime}(\bz)=-\Omega^{\prime}(\bz)$ and $\upsilon^{\prime\prime}(\bz)=-\Omega^{\prime\prime}(\bz)$. In the following section, we provide some examples of calculation of the main ingredients of the Newton-Raphson algorithm (e.g., the first and second derivative of the norm, $\Omega^{\prime}(\bz)$ and $\Omega^{\prime\prime}(\bz)$) for the numerical optimization with barrier, in the special cases of the $\ell_2$ and group $\ell_2$ norms. For those norms, the dual solution is available analytically. For the $\ell_1$-norm, $\Omega_1(\bx)=\Vert\bx\Vert_1=\sum_{j=1}^p\vert x_j\vert$ which corresponds to the LASSO penalty \citep[see][]{tibshirani1996}, the dual solution exist $\Omega_1^\ast(\bx)=\Vert\bx\Vert_\infty$, for any $\bx\in\mathbb{R}^p$. However, in such case the second derivative of the norm is not available since $\Omega_1(\bx)\in\mathcal{C}^1$ and the Newton-Raphson algorithm fails.
%
\subsubsection{Examples: $\ell_2$ and group $\ell_2$ norms}
\label{sec:example_analytical_norms}
%
\noindent Let $\bx\in\mathbb{R}^p$, for the $\ell_2$-norm $\Omega_2(\bx)=\Vert\bx\Vert_2$ which corresponds to the RIDGE penalty of \cite{hoerl_kennard.2000}, we have:
\begin{align}
\Omega_2^{\prime}(\bx)&=\frac{\bx}{\Vert\bx\Vert_2}\\
\Omega_2^{\prime\prime}(\bx)&=-\frac{\bx\bx^\intercal}{\big(\Vert\bx\Vert_2\big)^3}+\frac{1}{\Vert\bx\Vert_2}\bI_p.
\end{align}
For the Group $\ell_2$-norm which corresponds to the Group-LASSO penalty of \cite{yuan_lin.2006}, we have $\Omega_G(\bx)=\sum_{g=1}^G\sqrt{w_g}\Vert \bx_g\Vert_2$, where $\bx=\{\bx_{(1)},\dots,\bx_{(G)}\}$ and each $\bx_{(g)}$ for $1\leq g\leq G$ represents a group from $\bx$ and $w_g$ represents the number of weights in $\bx_g$ for $1\leq g\leq G$. Therefore:
\begin{align}
\Omega_G^{\prime}(\bx)&=\begin{bmatrix}\frac{\sqrt{w_1}\bx_1}{\Vert\bx_1\Vert_2}\\ \vdots\\\frac{\sqrt{w_G}\bx_G}{\Vert\bx_G\Vert_2}\end{bmatrix}\\
\Omega_G^{\prime\prime}(\bx)&=\mathsf{diag}\big(\sqrt{w_1}\Omega_2^{\prime\prime}(\bx_1),\dots,\sqrt{w_G}\Omega_2^{\prime\prime}(\bx_G)\big).
\end{align}
Both the $\ell_2$-norm $\Omega_2(\bx)$ and the group $\ell_2$-norm $\Omega_G(\bx)$ admit an analytical closed form expressions for the corresponding dual norms $\Omega_2^\ast(\bx)$ and $\Omega_G^\ast(\bx)$, therefore optimization methods for finding the numerical solution of the optimization problem in equations \eqref{eq:dual_norm_opt_prob_def_1}-\eqref{eq:dual_norm_opt_prob_def_2} are not required for these norms. In the following section we introduce a Newton-Raphson-type algorithm for dealing with the problem of finding the dual norm of the group $\ell_2$-norm with overlapping groups. The analytical solution for the overlap-group $\ell_2$-norm is not available, \citep[see, e.g.][]{jenatton2011structured}.
%
\subsubsection{Overlap group $\ell_2$-norm}
\label{sec:ovg_norm}
%
\noindent Let $\bx\in\mathbb{R}^{p}$, then the overlap-group-LASSO norm \citep[][]{jenatton2011structured}, is defined as
\begin{equation}
\Omega_{OG}(\bx)=\sum_{g=1}^{G}\Vert \bw_g\odot\bx\Vert_2,
\end{equation}
where $\bw_g=\widetilde{\bw}\odot\bs_g\in\mathbb{R}^{p}$ and $\bs_g\in\mathbb{R}^{p}$ is a group selection vector that selects the elements of $\bx$ belonging to the $g$-th group in such a way that $\bx_g=\bs_g^\intercal\bx$ for $g=1,\dots,G$ and $\bS\in\mathbb{R}^{G\times p}$ is a group selection matrix obtained by collecting the group selection vectors by row, i.e. $\bS=\big(\bs_1, \bs_2 , \dots , \bs_G\big)^\intercal$. Unlike the group-LASSO, $\widetilde{\bw}\in\mathbb{R}^p$ is a vector that specifies the weight assigned to each element of $\bx$ and it is not group specific. Specifically, for $1\leq g\leq G$, $\bs_g=\big(s_1^{(g)},\dots,s_p^{(g)}\big)^\intercal\in\mathbb{R}^p$ has general entry $s^{(g)}_{l}$ defined as $1$ if the $l$-th component of $\bx$, $x_{l}$, belongs to group $g$ and $0$ otherwise and $\widetilde{\bw}=\big(\widetilde{w}_1,\dots,\widetilde{w}_p\big)^\intercal$ with general entry
$\widetilde{w}_l=\big(\sum_{g=1}^{G}s^{(g)}_{l}\big)^{-1}$ for $l=1,\dots,p$. The previous overlap-group $\ell_2$-norm can be easily rewritten as 
\begin{equation}
\label{eq:ovg_norm_definition}
\Omega_{OG}(\bz)=\sum_{g=1}^{G}\Vert \bc_g\odot\bx_g\Vert_2,
\end{equation}
where $\bx_g=(x_{g,1},\dots,x_{g,n_g})=\bs_g^\intercal\bx\in\mathbb{R}^{n_g}$ has $g$-th entry corresponding to the elements of $\bx\in\mathbb{R}^{p}$ that belong to the $g$-th group, $\bc_g=\big(c_{g,1},\dots,c_{g,n_g})\in\mathbb{R}^{n_g}$ with $n_g=\sum_{l=1}^p s_{l}^{(g)}$ and $\bc_{g}=\widetilde\bw\big[\bbone_{[0,\infty)}(\widetilde\bw)\big]$ where $\bbone_{[0,\infty)}(\widetilde{\bw}_g)=\begin{bmatrix}\bbone_{[0,\infty)}(\widetilde{w}_{1}) & \cdots & \bbone_{[0,\infty)}(\widetilde{w}_{p})\end{bmatrix}$ denotes the element-wise indicator function, for $g=1,\dots,G$. When some of the groups overlap, the penalty $\Omega_{OG}(\bx)$ is still a norm (if all covariates are in at least one group) whose ball has singularities when some $\bx_g$ are equal to zero, (which is not our case).\par 
For the overlap-group $\ell_2$-norm defined in equation \eqref{eq:ovg_norm_definition}, we have
\begin{align}
\Omega_{OG}^{\prime}(\bx)&=\begin{bmatrix}x_1\widetilde{w}_{1}^2\sum_{g=1}^G \frac{1}{\Vert \bc_g\odot\bx_g\Vert_2}\\ \vdots\\
x_p
\widetilde{w}_{p}^2\sum_{g=1}^G \frac{1}{\Vert \bc_g\odot\bx_g\Vert_2}\end{bmatrix},
\end{align}
and
\begin{align}
\Omega_G^{\prime\prime}(\bx)&=\begin{bmatrix}
\Omega_{OG,(1,1)}^{\prime\prime}(\bx)&\Omega_{OG,(1,2)}^{\prime\prime}(\bx)&\cdots &\Omega_{OG,(1,p)}^{\prime\prime}(\bx)\\
\Omega_{OG,(2,1)}^{\prime\prime}(\bx)&\Omega_{OG,(2,2)}^{\prime\prime}(\bx)&\cdots &\Omega_{OG,(2,p)}^{\prime\prime}(\bx)\\
\vdots&\vdots&\ddots&\vdots\\
\Omega_{OG,(p,1)}^{\prime\prime}(\bx)&\Omega_{OG,(p,2)}^{\prime\prime}(\bx)&\cdots &\Omega_{OG,(p,p)}^{\prime\prime}(\bx)\\
\end{bmatrix},
\end{align}
where
\begin{align}
\Omega_{OG,(j,j)}^{\prime\prime}(\bx)&=\widetilde{w}_{j}^2\Bigg(\sum_{g=1}^G \frac{1}{\Vert \bc_g\odot\bx_g\Vert_2} -x_j^2\widetilde{w}_{j}^2\sum_{g=1}^G \frac{1}{\big(\Vert \bc_g\odot\bx_g\Vert_2\big)^3}\Bigg)\\
\Omega_{OG,(j,k)}^{\prime\prime}(\bx)&=x_jx_k\widetilde{w}_{j}^2\widetilde{w}_{k}^2\sum_{g=1}^G \frac{1}{\big(\Vert \bc_g\odot\bx_g\Vert_2\big)^3}.
\end{align}
Note that the hessian matrix is singular if and only if there exists a $j\in\{1,\dots,p\}$ such that $x_j$ does not belong to any group, e.g., $g\notin j$ for $g=1,\dots,G$.
%
\section{MM algorithm for barrier methods}
\label{sec:mm_algo}
%
Now the challenge is to find an efficient way to optimize $\mathcal{L}_\varrho(\mathbf{z})$. Fortunately, 
$\mathcal{J}(\mathbf{z})$ is concave, then it can be minorized by a linear function, making so possible to implement a minorization-maximization (MM) 
scheme. Before proceeding let us briefly introduce the MM algorithm and its basic properties.\newline

\noindent Suppose we want to minimise the objective function $\mathcal{L}_\varrho\big(\mathbf{z}\big):\Re^{p+1}\to\Re$ and denote with $\widehat{\mathbf{z}}^{(k)}$ the current iterate, the MM algorithm proceeds in two steps:
\begin{itemize}
\item[{\it (i)}] create a surrogate (majorizer) function $g\big(\mathbf{z}\vert\widehat{\mathbf{z}}^{(k)}\big):\mathbb{R}^p\times\mathbb{R}^p\to\mathbb{R}$ that satisfies
\begin{enumerate}
\item[{\it (i.1)}] $g\big(\widehat{\mathbf{z}}^{(k)}\vert\widehat{\mathbf{z}}^{(k)}\big)=\mathcal{L}_\varrho\big(\widehat{\mathbf{z}}^{(k)}\big)$;
\item[{\it (i.2)}] $g\big(\mathbf{z}\vert\widehat{\mathbf{z}}^{(k)}\big)\geq \mathcal{L}_\varrho\big(\mathbf{z}\big)$, for all $\mathbf{z}$;
\end{enumerate}
\item[{\it (ii)}] $\widehat{\mathbf{z}}^{(k+1)}=\arg\min_{\mathbf{z}\in\Re^{p+1}}g\big(\mathbf{z}\vert\widehat{\mathbf{z}}^{(k)}\big)$.
\end{itemize}
Therefore, $g\big(\mathbf{z}^{(k+1)}\vert\vartheta^{(k)}\big)\leq g\big(\widehat{\mathbf{z}}^{(k)}\vert\widehat{\mathbf{z}}^{(k)}\big)$, with conditions {\it (i)} and {\it (ii)} entails the descent property, i.e., $\mathcal{L}_\varrho\big(\mathbf{z}^{(k+1)}\big)\leq \mathcal{L}_\varrho\big(\widehat{\mathbf{z}}^{(k)}\big)$. For further references we refer to \cite{lange.2010}, \cite{hunter_lange.2004}, \cite{lange_etal.2014} and to \cite{wu_etal.2010} for a comparison between EM and MM.\par
One way of improving the barrier method is to change the barrier constant as the iterations proceed. This sounds vague, but matters simplify enormously if we view the construction of an adaptive barrier method from the perspective of the MM algorithm. Proposition \ref{prop:mm} provide the analytic expression for the majorizer function 
$\mathcal{L}_\varrho(\mathbf{z} \vert \widehat{\mathbf{z}}^{(k)})$, that is the basic ingredient to derive our MM update.
\begin{proposition}
\label{prop:mm}
Consider the Lagrangian functions $\mathcal{L}_\varrho(\mathbf{z})$ as given in equation \eqref{eq:L}, and define
\begin{equation}
g(\bz\vert\widehat{\mathbf{z}}^{(k)})=-\bz^\intercal\bx-\varrho\upsilon(\widehat{\mathbf{z}}^{(k)})\log(\upsilon(\bz))+\varrho\upsilon^\prime(\widehat{\mathbf{z}}^{(k)})(\bz-\widehat{\mathbf{z}}^{(k)}),
\end{equation}
where $\upsilon^\prime(\widehat{\mathbf{z}}^{(k)})$ is the gradient of the function $\upsilon(\bz)$ evaluated at $\widehat{\mathbf{z}}^{(k)}$ defined in equation \eqref{eq:dual_norm_opt_prob_def_2}. Then $\mathcal{L}_\varrho(\mathbf{z} \vert \widehat{\mathbf{z}}^{(k)})$ majorizes $\mathcal{L}_\varrho(\mathbf{z})$, which means that for any pair $(\mathbf{z}, \widehat{\mathbf{z}}^{(k)})$ we have $\mathcal{L}_\varrho(\mathbf{z}) \leq \mathcal{L}_\varrho(\mathbf{z} \vert \widehat{\mathbf{z}}^{(k)})$, and $\mathcal{L}_\varrho(\mathbf{z}) = \mathcal{L}_\varrho(\mathbf{z} \vert \widehat{\mathbf{z}}^{(k)})$ if and only if $\mathbf{z} = \widehat{\mathbf{z}}^{(k)}$.
\end{proposition}
\begin{proof}
See \cite{lange.2013}. The proof is reported here for the sake of completeness
Consider the following inequalities
\begin{align}
&-\upsilon(\widehat{\mathbf{z}}^{(k)})\log(\upsilon(\bz))+\upsilon(\widehat{\mathbf{z}}^{(k)})\log(\upsilon(\widehat{\mathbf{z}}^{(k)}))+\upsilon^\prime(\widehat{\mathbf{z}}^{(k)})(\bz-\widehat{\mathbf{z}}^{(k)})\nonumber\\
\geq&-\frac{\upsilon(\widehat{\mathbf{z}}^{(k)})}{\upsilon(\widehat{\mathbf{z}}^{(k)})}\big(\upsilon(\bz)-\upsilon(\widehat{\mathbf{z}}^{(k)})\big)+\upsilon^\prime(\widehat{\mathbf{z}}^{(k)})(\bz-\widehat{\mathbf{z}}^{(k)})\nonumber\\
=&-\upsilon(\bz)+\upsilon(\widehat{\mathbf{z}}^{(k)})+\upsilon^\prime(\widehat{\mathbf{z}}^{(k)})(\bz-\widehat{\mathbf{z}}^{(k)})\nonumber\\
\geq&0,
\end{align}
based on the concavity of the functions $\log(y)$ and $\upsilon(\bz)$. Because equality holds throughout when $\bz = \widehat{\mathbf{z}}^{(k)}$, we have identified a novel function majorizing $0$ and incorporating a barrier for $\upsilon(\bz)$. The significance of this discovery is that the surrogate function
\begin{equation}
\label{eq:surrogate}
g(\bz\vert\widehat{\mathbf{z}}^{(k)})=-\bz^\intercal\bx-\varrho\upsilon(\widehat{\mathbf{z}}^{(k)})\log(\upsilon(\bz))+\varrho\upsilon^\prime(\widehat{\mathbf{z}}^{(k)})(\bz-\widehat{\mathbf{z}}^{(k)}),
\end{equation}
majorizes $\mathcal{L}_\varrho(\bz)$ up to an irrelevant additive constant. Minimization of the surrogate function drives $\mathcal{L}_\varrho(\bz)$ downhill while keeping the inequality constraints inactive. In the limit, one or more of the inequality constraints may become active.
\end{proof}
\noindent As described by, for instance, \citet{lange.2013}, \citet{lange.2016} and \citet{sun_etal.2017}, if $\mathcal{L}_\varrho(\mathbf{z} \vert \widehat{\mathbf{z}}^{(k)})$ is a 
majorizer of the convex function $\mathcal{L}_\varrho(\mathbf{z})$, then the sequence $\{\widehat{\mathbf{z}}^{(k)}\}$, defined through the recursion formula 
$\widehat{\mathbf{z}}^{(k+1)} = \arg\max_{\mathbf{z}} \mathcal{L}_\varrho(\mathbf{z} \vert \widehat{\mathbf{z}}^{(k)})$, converges to the global minimum of $\mathcal{L}_\varrho(\mathbf{z})$. 
Unfortunately, in our case, the optimization step of the surrogate function $g(\bz\vert\widehat{\mathbf{z}}^{(k)})$ in equation \eqref{eq:surrogate} does not admit a closed form expression, therefore, we must revert to the MM gradient algorithm. The Newton-Raphson's update for a fixed length $\gamma>0$ is 
\begin{equation}
\label{eq:newton_update}
\widehat{\mathbf{z}}^{(k+1)}=\widehat{\mathbf{z}}^{(k)}-\gamma\big(d^2g(\widehat{\mathbf{z}}^{(k)}\vert\widehat{\mathbf{z}}^{(k)})\big)^{-1}dg(\widehat{\mathbf{z}}^{(k)}\vert\widehat{\mathbf{z}}^{(k)}),
\end{equation}
and requires the first and second differentials
\begin{align}
dg(\widehat{\mathbf{z}}^{(k)}\vert\widehat{\mathbf{z}}^{(k)})&=df(\widehat{\mathbf{z}}^{(k)})\\
d^2g(\widehat{\mathbf{z}}^{(k)}\vert\widehat{\mathbf{z}}^{(k)})&=d^2f(\widehat{\mathbf{z}}^{(k)})-\varrho \upsilon^{\prime\prime}(\widehat{\mathbf{z}}^{(k)})+\frac{\varrho}{\upsilon(\widehat{\mathbf{z}}^{(k)})}\upsilon^\prime(\widehat{\mathbf{z}}^{(k)})\big(\upsilon^\prime(\widehat{\mathbf{z}}^{(k)})\big)^\intercal,
\end{align}
where $\upsilon^\prime(\widehat{\mathbf{z}}^{(k)})-\Omega
^\prime(\widehat{\mathbf{z}}^{(k)})$ and $\upsilon^{\prime\prime}(\widehat{\mathbf{z}}^{(k)})-\Omega^{\prime\prime}(\widehat{\mathbf{z}}^{(k)})$ in Section \ref{sec:barrier_methods}. Algorithm \ref{alg:dual_norm_MM} provides a short description of the main steps required by the MM method for solving the constrained optimization problem in equation \eqref{eq:copt}.\newline
\indent Figure \ref{fig:sim1_results} provides the results of two simulation experiments for the $\ell_2$-norm (\ref{fig:norm2}) and the group $\ell_2$-norm (\ref{fig:Gnorm2}). For the $\ell_2$-norm $\Omega_2(\bx)$ we simulated $n=50$ vectors $\bx\in\mathbb{R}^p$ of varying dimensions $p=(5, 10, 50, 100, 200, 500)$ and Figure \eqref{fig:norm2} provides the boxplots of the difference between the analytical (true) dual norm $\Omega^\ast_2(\bx)=\Omega_2(\bx)$ and the value obtained by running algorithm \ref{alg:dual_norm_MM}. For the group $\ell_2$-norm we instead considered groups of dimensions $p_G=(2,5,8,10)$ and let also vary the number of groups $n_G=(2,5,10,20)$. Therefore we simulated $n=50$ vectors $\bx\in\mathbb{R}^p$ of varying dimensions $p=p_G\times n_G$ and Figure \eqref{fig:Gnorm2} provides the boxplots of the difference between the analytical (true) dual norm $\Omega^\ast_G(\bx)=\max_{g=1,\dots,G}\frac{1}{\sqrt{w_g}}\Vert \bx_g\Vert_2$ and the value obtained by running algorithm \ref{alg:dual_norm_MM}.
%
\begin{algorithm*}[t!]
\begingroup
    \fontsize{10pt}{11pt}\selectfont
\caption{MM algorithm for dual norm evaluation} 
 \label{alg:dual_norm_MM}
 {Initialise the regression parameters: $\widehat{\bz}^{(0)}$ and evaluate the objective function $\mathcal{L}_\varrho\big(\widehat{\bz}^{(0)}\big)$ and set $k=0$ and set $\mathsf{maxiter}$, $\epsilon>0$, and $\varrho\gg0$}\;   
 \While(){$\big\vert\mathcal{L}_\varrho\big(\widehat{\bz}^{(k)}\big)-\mathcal{L}_\varrho\big(\widehat{\bz}^{(k-1)}\big)\big\vert>\epsilon$ and $k>1$}
 {
 \For(){$j$ \mbox{from} $1$ to $\mathsf{maxiter}$}{
{update $\widehat{\bz}^{(k+1)}$}, using the Newton-Raphson update in equation \eqref{eq:newton_update}\;
}
{evaluate $\mathcal{L}_\varrho\big(\widehat{\bz}^{(k+1)}\big)$}\;
}
\endgroup
\end{algorithm*}
%

\begin{figure}[!ht]
\begin{center}
\subfloat[$\ell_2$-norm]{\label{fig:norm2}\includegraphics[width=0.7\textwidth]{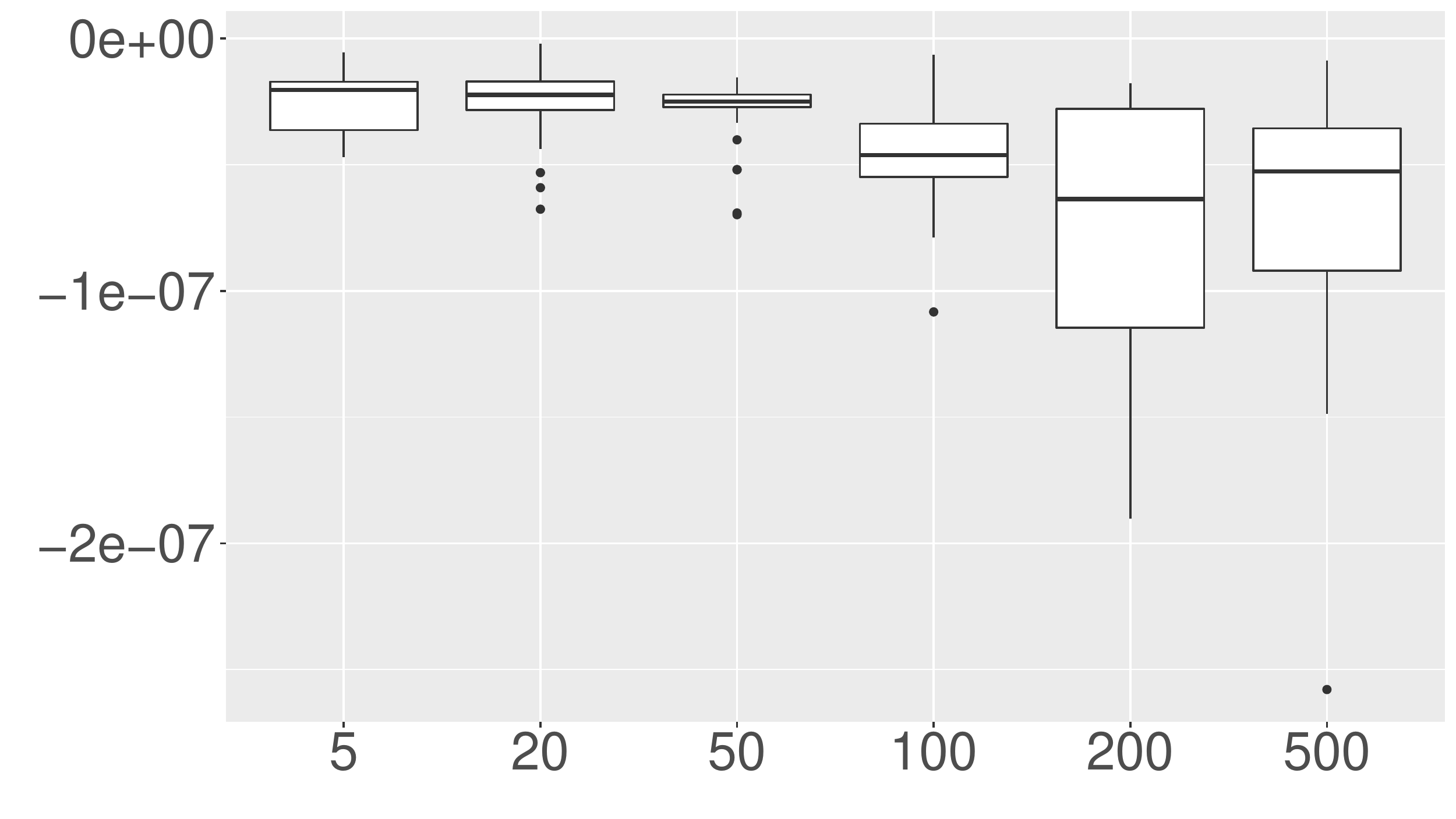}}\qquad
\subfloat[group $\ell_2$-norm]{\label{fig:Gnorm2}\includegraphics[width=0.7\textwidth]{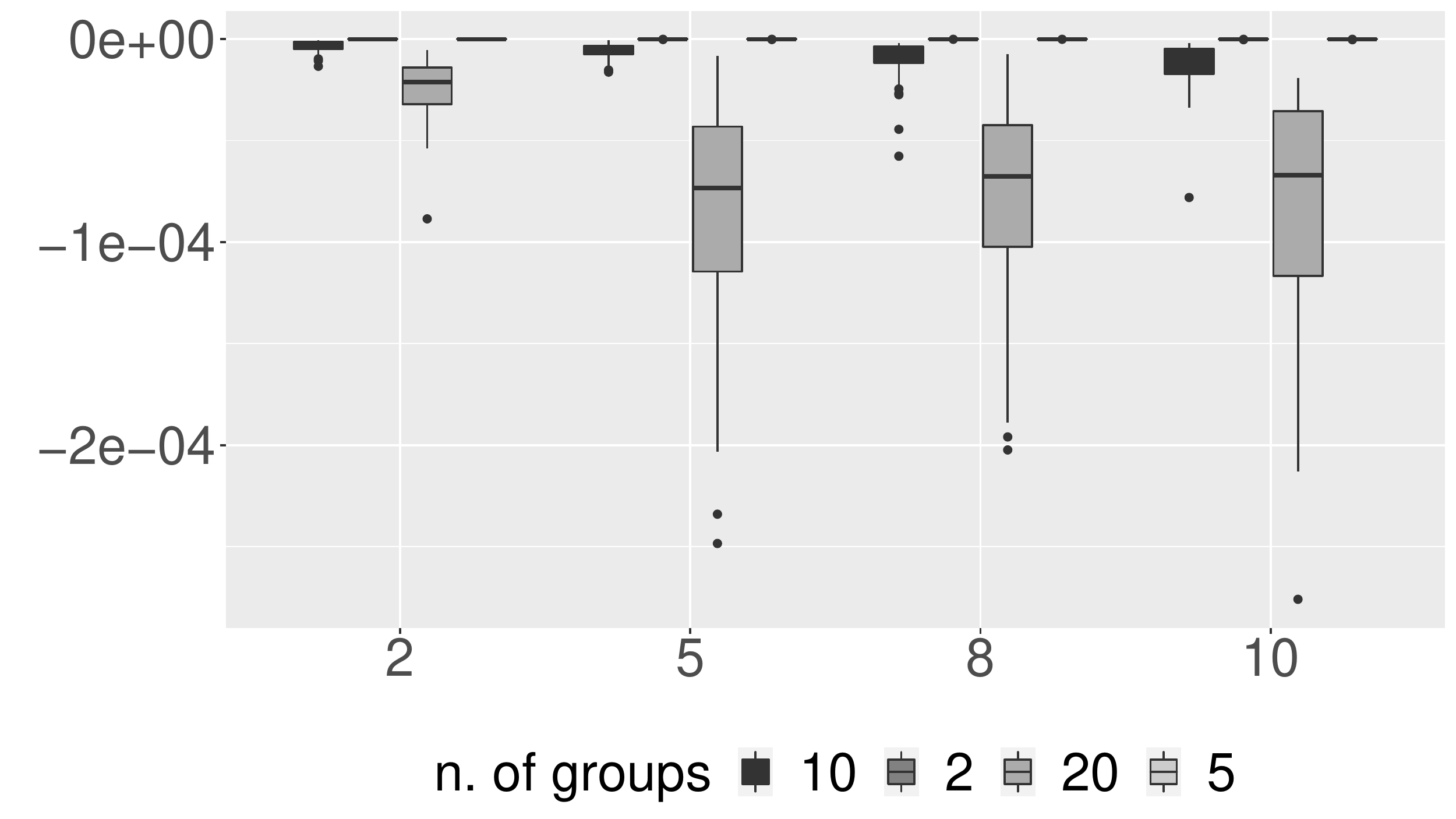}}
\caption{Box plots of the difference between the analytical dual norm and that obtained by running Algorithm \ref{alg:dual_norm_MM} over $50$ replications of random vectors of varying dimension (for the $\ell_2$-norm) and varying group-length and number of groups (for the group $\ell_2$-norm).}
\label{fig:sim1_results}
\end{center}
\end{figure}

\begin{figure}[!ht]
\begin{center}
\includegraphics[width=0.7\textwidth]{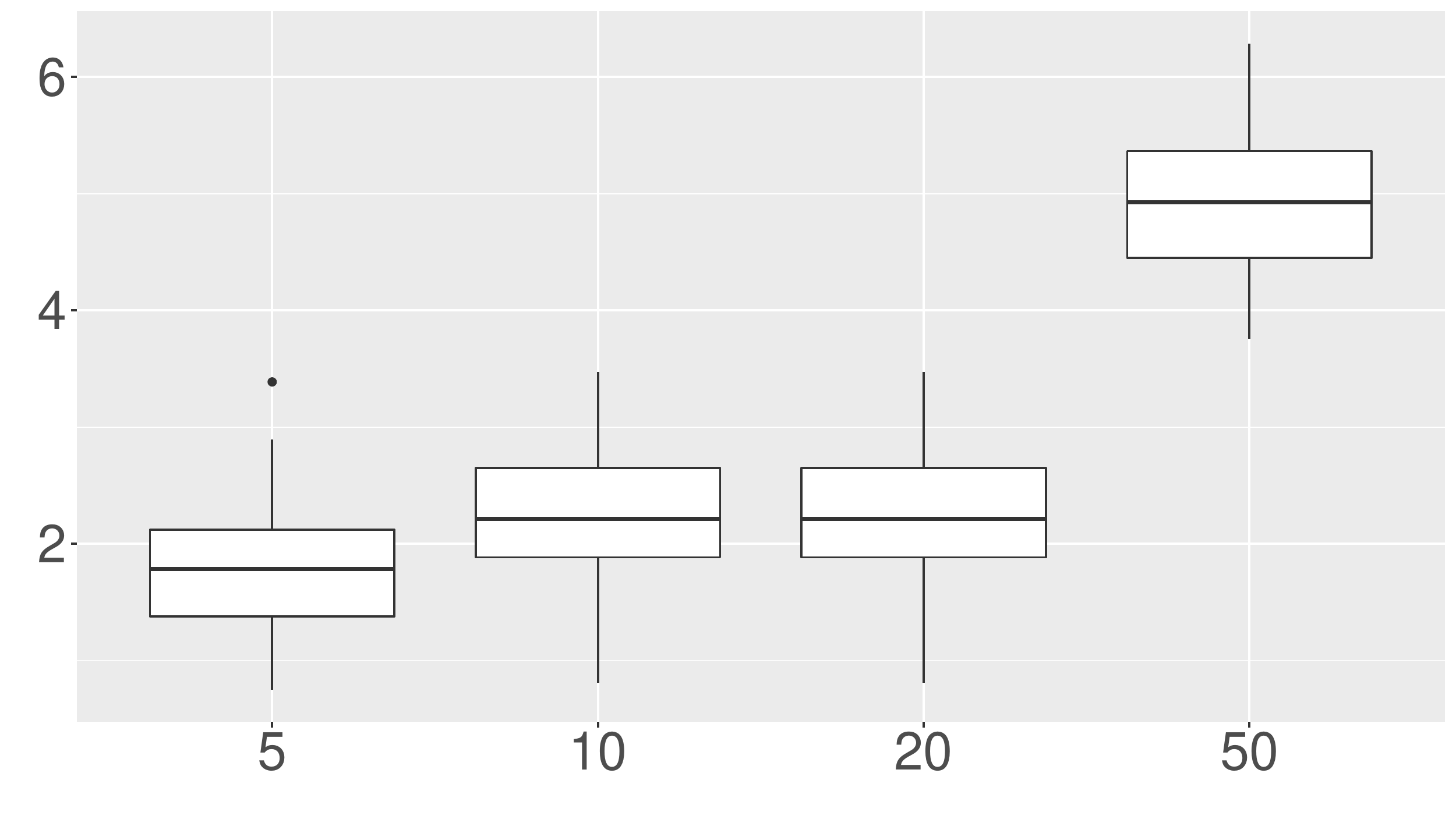}
\caption{Box plots of the difference between the $\Omega_{OG}(\bx)$ dual norm obtained by running Algorithm \ref{alg:dual_norm_MM} over $50$ replications of random vectors of varying dimension and varying group-length and number of groups.}
\label{fig:sim2_results}
\end{center}
\end{figure}

%
\section{Application}
\label{sec:appl}
%
\noindent In this application we consider the dual norm of overlap group $\ell_2$-norm, $\Omega_{OG}(\bx)$ introduced by \cite{bernardi_etal.2021} for kernel penalization in the functional regression context. The norm is defined as
\begin{equation}
\Omega^c_{\cal I}( u_{\mathcal{I}^c} ) = \sum_{b \in {\cal B}^c}  \Vert c_{b}^{({\cal I}^c)} \odot u_{\mathcal{I}^c} \Vert_2 ,
\end{equation}
where ${\cal B }\subset \{1, \dots, B, B+1\}$ is the set of indices of the true non-zero groups and ${\cal I} \subset \{(m,l): m=1, \dots, M; \, l = 1, \dots, L\}$ is the set of the indices of the non-zero coefficients $\psi_{m,l}$, $c_b =\mathsf{vec}(  \mathbf{S}_b \odot \mathbf{C})$, 
where $  \mathbf{S}_b$ is the $M \times L$ selection matrix with general entry $s_{m,l}^{(b)}$, defined as $1$ if the parameter $\psi_{m,l}$ belongs to group $b$ and 0 otherwise and $\mathbf{C}$ is the matrix with general element $c_{m,l}$ defined as $c_{ml} = \big( \sum_{b=1}^{B+1}  s_{m,l}^{(b)} \big)^{-1}$.\par
To check for the irrepresentable condition condition, we need to calculate the dual norm $\big(\Omega^c_{\cal B} \big)^*\big( {\mathbf{Z}}_0^\intercal {\mathbf{Z}}_1 ({\mathbf{Z}}_1^\intercal {\mathbf{Z}}_1)^{-1} \mathbf{r}_1  \big) $, where $\mathbf{r}_1$ is the vector containing, for $(m,l) \in \cal I$, the elements 
$\psi_{m,l} c^2_{m,l} \sum_{b  \in {\cal B}, b: \psi_b \ni \psi_{m,l} } \Vert c_b \odot \bpsi \Vert_2^{-1}$.
%
\section{Conclusion}
\label{sec:conclu}
%
\noindent In this paper we provide a new algorithm for solving numerically the dual norm optimization problem. The methods relies on the Majorization-Minimization principle of \cite{lange.2016} combined with barrier methods that efficiently convert a constrained optimization problem into an unconstrained one. Extensive simulation experiments have been performed in order to verify the correctness of operation, and evaluate the performance of the proposed method. Our results demonstrate the effectiveness of the algorithm in retrieving the dual norm even for large dimensions.

%
\bibliographystyle{apalike}
\bibliography{refs.bib}

\begin{thebibliography}{}

\bibitem[Bach et~al., 2012]{bach_etal.2012}
Bach, F., Jenatton, R., Mairal, J., and Obozinski, G. (2012).
\newblock Optimization with sparsity-inducing penalties.
\newblock {\em Foundations and Trends® in Machine Learning}, 4(1):1--106.

\bibitem[Bernardi et~al., 2021]{bernardi_etal.2021}
Bernardi, M., Canale, A., and Stefanucci, M. (2021).
\newblock Locally sparse function on function regression.

\bibitem[Boyd et~al., 2011]{boyd_etal.2011}
Boyd, S., Parikh, N., Chu, E., Peleato, B., and Eckstein, J. (2011).
\newblock Distributed optimization and statistical learning via the alternating
  direction method of multipliers.
\newblock {\em Foundations and Trends{\textregistered} in Machine Learning},
  3(1):1--122.

\bibitem[B\"{u}hlmann and van~de Geer, 2011]{buhlmann_van_de_geer.2011}
B\"{u}hlmann, P. and van~de Geer, S. (2011).
\newblock {\em Statistics for high-dimensional data}.
\newblock Springer Series in Statistics. Springer, Heidelberg.
\newblock Methods, theory and applications.

\bibitem[Hoerl and Kennard, 2000]{hoerl_kennard.2000}
Hoerl, A.~E. and Kennard, R.~W. (2000).
\newblock Ridge regression: Biased estimation for nonorthogonal problems.
\newblock {\em Technometrics}, 42(1):80--86.

\bibitem[Hunter and Lange, 2004]{hunter_lange.2004}
Hunter, D.~R. and Lange, K. (2004).
\newblock A tutorial on {MM} algorithms.
\newblock {\em Amer. Statist.}, 58(1):30--37.

\bibitem[Jenatton et~al., 2011]{jenatton2011structured}
Jenatton, R., Audibert, J.-Y., and Bach, F. (2011).
\newblock Structured variable selection with sparsity-inducing norms.
\newblock {\em The Journal of Machine Learning Research}, 12:2777--2824.

\bibitem[Lange, 2010]{lange.2010}
Lange, K. (2010).
\newblock {\em Numerical analysis for statisticians}.
\newblock Statistics and Computing. Springer, New York, second edition.

\bibitem[Lange, 2013]{lange.2013}
Lange, K. (2013).
\newblock {\em Optimization}, volume~95 of {\em Springer Texts in Statistics}.
\newblock Springer, New York, second edition.

\bibitem[Lange, 2016]{lange.2016}
Lange, K. (2016).
\newblock {\em M{M} optimization algorithms}.
\newblock Society for Industrial and Applied Mathematics, Philadelphia, PA.

\bibitem[Lange et~al., 2014]{lange_etal.2014}
Lange, K., Chi, E.~C., and Zhou, H. (2014).
\newblock A brief survey of modern optimization for statisticians.
\newblock {\em International Statistical Review}, 82(1):46--70.

\bibitem[Luenberger and Ye, 2021]{luenberger_yinyu.2021}
Luenberger, D.~G. and Ye, Y. ([2021] \copyright 2021).
\newblock {\em Linear and nonlinear programming}, volume 228 of {\em
  International Series in Operations Research \& Management Science}.
\newblock Springer, Cham, fifth edition.

\bibitem[Nocedal and Wright, 2006]{nocedal_wright.2006}
Nocedal, J. and Wright, S.~J. (2006).
\newblock {\em Numerical optimization}.
\newblock Springer Series in Operations Research and Financial Engineering.
  Springer, New York, second edition.

\bibitem[Osborne et~al., 2000]{osborne_etal.2000}
Osborne, M.~R., Presnell, B., and Turlach, B.~A. (2000).
\newblock On the {LASSO} and its dual.
\newblock {\em J. Comput. Graph. Statist.}, 9(2):319--337.

\bibitem[Sun et~al., 2017]{sun_etal.2017}
Sun, Y., Babu, P., and Palomar, D.~P. (2017).
\newblock Majorization-minimization algorithms in signal processing,
  communications, and machine learning.
\newblock {\em IEEE Transactions on Signal Processing}, 65(3):794--816.

\bibitem[Tibshirani, 1996]{tibshirani1996}
Tibshirani, R. (1996).
\newblock Regression shrinkage and selection via the lasso.
\newblock {\em Journal of the Royal Statistical Society: Series B
  (Methodological)}, 58(1):267--288.

\bibitem[Wu and Lange, 2010]{wu_etal.2010}
Wu, T.~T. and Lange, K. (2010).
\newblock The {MM} alternative to {EM}.
\newblock {\em Statist. Sci.}, 25(4):492--505.

\bibitem[Yuan and Lin, 2006]{yuan_lin.2006}
Yuan, M. and Lin, Y. (2006).
\newblock Model selection and estimation in regression with grouped variables.
\newblock {\em J. R. Stat. Soc. Ser. B Stat. Methodol.}, 68(1):49--67.

\bibitem[Zhao and Yu, 2006]{zhao_yu.2006}
Zhao, P. and Yu, B. (2006).
\newblock On model selection consistency of lasso.
\newblock {\em Journal of Machine Learning Research}, 7(90):2541--2563.

\bibitem[Zou and Hastie, 2005]{zou_hastie.2005}
Zou, H. and Hastie, T. (2005).
\newblock Addendum: ``{R}egularization and variable selection via the elastic
  net'' [{J}. {R}. {S}tat. {S}oc. {S}er. {B} {S}tat. {M}ethodol. {\bf 67}
  (2005), no. 2, 301--320; mr2137327].
\newblock {\em J. R. Stat. Soc. Ser. B Stat. Methodol.}, 67(5):768.

\end{thebibliography}

\end{document}